%% file: arxiv_main_november.tex
\documentclass{article} 
\usepackage{balance} 
\usepackage{natbib}
\usepackage{tikz}
\usepackage{amsmath}
\usepackage[margin = 1.25in]{geometry}
\usepackage{graphicx}
\DeclareMathOperator*{\argmax}{arg\,max}
\DeclareMathOperator*{\argmin}{arg\,min}

\input{00preamble}

\title{Relaxations of Envy-Freeness Over Graphs}
\author{Justin Payan, Rik Sengupta and Vignesh Viswanathan \\ University of Massachusetts Amherst \\ \texttt{\{jpayan, rsengupta, vviswanathan\}@umass.edu}}

\date{}

\begin{document}




\maketitle 


\begin{abstract}
When allocating a set of indivisible items among agents, the ideal condition of {\em envy-freeness} cannot always be achieved.
\emph{Envy-freeness up to any good} (EFX), and \emph{envy-freeness with $k$ hidden items} (HEF-$k$) are two very compelling relaxations of envy-freeness, which remain elusive in many settings. 
We study a natural relaxation of these two fairness constraints, where we place the agents on the vertices of an undirected graph, and only require that our allocations satisfy the EFX (resp. HEF) constraint on the edges of the graph. 
We refer to these allocations as {\em graph}-EFX (resp. {\em graph}-HEF) or simply $G$-EFX (resp. $G$-HEF) allocations.
We show that for any graph $G$, there always exists a $G$-HEF-$k$ allocation of goods, where $k$ is the size of a minimum vertex cover of $G$, and that this is essentially tight.
We show that $G$-EFX allocations of goods exist for three different classes of graphs --- two of them generalizing the star $K_{1, n-1}$ and the third generalizing the three-edge path $P_4$. Many of these results extend to allocations of \emph{chores} as well. Overall, we show several natural settings in which the graph structure helps obtain strong fairness guarantees.
Finally, we evaluate an algorithm using problem instances from Spliddit to show that $G$-EFX allocations appear to exist for paths $P_n$, pointing the way towards showing EFX for even broader families of graphs.
\end{abstract}

\input{arxiv_01}

\input{arxiv_02}

\input{arxiv_03}

\input{arxiv_04}

\input{arxiv_05}

\input{arxiv_06}

\input{arxiv_07}

\section*{Acknowledgments}
We thank Raghav Addanki for participating in early discussions of this work. We also thank Rohit Vaish for pointing us to several interesting concepts in fair allocation, and suggesting approaches towards some of the proof techniques we used. We thank Yair Zick and Hadi Hosseini for helpful comments and suggestions. Finally, we thank Andrew McGregor in pointing us to the real-world characterization of graphs with large diameter. RS acknowledges support from NSF grant no. CCF-1908849.

\bibliographystyle{plainnat} 
\bibliography{abb, references}

\input{arxiv_appendix}

\end{document}

%% file: 00preamble.tex
\usetikzlibrary{matrix,decorations.pathreplacing}

\usepackage{xcolor}

\usepackage{paralist}
\usepackage{easybmat}
\usepackage{multirow,bigdelim}
\usepackage{amsmath}
\usepackage{amsthm}

\usepackage{amssymb}
\usepackage{algorithm}
\usepackage[justification = centering]{subcaption}
\usepackage{color}
\usepackage[english]{babel}
\usepackage{graphicx}
\usepackage{wrapfig,epsfig}
\usepackage{epstopdf}
\usepackage{url}
\usepackage{graphicx}
\usepackage{color}
\usepackage{epstopdf}
\usepackage[noend]{algpseudocode}
\usepackage[T1]{fontenc}
\usepackage{bbm}
\usepackage{comment}
\usepackage{xspace}
\usepackage{thmtools}

\usepackage{tikz}
\usepackage{hyperref}  
\hypersetup{colorlinks=true,citecolor=blue,linkcolor=blue} 
\usetikzlibrary{arrows}

\def\@centernot#1#2{%
  \mathrel{%
    \rlap{%
      \settowidth\dimen@{$\m@th#1{#2}$}%
      \kern.5\dimen@
      \settowidth\dimen@{$\m@th#1=$}%
      \kern-.5\dimen@
      $\m@th#1\not$%
    }%
    {#2}%
  }%
}

\newtheorem{theorem}{Theorem}[section]

\newtheorem{definition}[theorem]{Definition}

\newtheorem{proposition}[theorem]{Proposition}

\newtheorem{observation}[theorem]{Observation}

\newtheorem{example}[theorem]{Example}

\newcommand{\envy}{\textit{envy}}
\newcommand{\efxenvy}{\textit{strong-envy}}

\newcommand{\N}{\mathcal{N}}
\newcommand{\R}{\mathbb{R}}

\renewcommand{\i}{\mathbf{i}}

\renewcommand{\varepsilon}{\epsilon}

\newcommand{\floor}[1]{\left\lfloor #1 \right\rfloor}
\newcommand{\ceil}[1]{\left\lceil #1 \right\rceil}

\definecolor{mygreen}{RGB}{80,180,0}
\definecolor{b2}{RGB}{51,153,255}
\definecolor{mycy2}{RGB}{255,51,255}

\makeatletter
\newcommand*{\RN}[1]{\expandafter\@slowromancap\romannumeral #1@}
\makeatother
\usepackage{lineno}

\input{libtheorems}

\newif\ifcomments
\commentstrue
\ifcomments
\newcommand{\rik}[1]{{\textcolor{red}{Rik: { #1}}}}
\else
\newcommand{\rik}[1]{}
\fi
\newif\ifcomments
\commentstrue
\ifcomments

\else
\newcommand{\raghav}[1]{}
\fi
\newif\ifcomments
\commentstrue
\ifcomments
\newcommand{\vignesh}[1]{{\textcolor{orange}{Vignesh: { #1}}}}
\else
\newcommand{\vignesh}[1]{}
\fi

%% file: libtheorems.tex
\usepackage{etoolbox}

\makeatletter

\newcommand{\define}[4][ignore]{%
  \ifstrequal{#1}{ignore}{}{
  \@namedef{thmtitle@#2}{#1}}%
  \@namedef{thm@#2}{#4}%
  \@namedef{thmtypen@#2}{lemma}%
  \newtheorem{thmtype@#2}[theorem]{#3}%
  \newtheorem*{thmtypealt@#2}{#3~\ref{#2}}%
}

\newcommand{\state}[1]{%
  \@namedef{curthm}{#1}
  \@ifundefined{thmtitle@#1}{
  \begin{thmtype@#1}
    }{
  \begin{thmtype@#1}[\@nameuse{thmtitle@#1}]
  }
    \label{#1}
    \@nameuse{thm@#1}
  \end{thmtype@#1}
  \@ifundefined{thmdone@#1}{
  \@namedef{thmdone@#1}{stated}%
  }{}
}

\newcommand{\restate}[1]{%
  \@namedef{curthm}{#1}
  \@ifundefined{thmtitle@#1}{
    \begin{thmtypealt@#1}
    }{
  \begin{thmtypealt@#1}[\@nameuse{thmtitle@#1}]
  }
    \@nameuse{thm@#1}
  \end{thmtypealt@#1}
  \@ifundefined{thmdone@#1}{
  \@namedef{thmdone@#1}{stated}%
  }{}
}

\newcommand{\thmlabel}[1]{
  \@ifundefined{thmdone@\@nameuse{curthm}}{\label{#1}
    }{\tag*{\eqref{#1}}}
}
\makeatother

%% file: arxiv_01.tex
\section{Introduction}\label{sec:intro}

The problem of fairly allocating a set of indivisible goods among agents with preferences has been extensively studied by the multi-agent systems community \citep{fairallocationsurvey}.

Several notions of fairness have been proposed and analyzed in the last two decades; of all these notions, arguably the most compelling one is that of {\em envy-freeness}. 
In an envy-free allocation of goods, no agent prefers the set of goods allocated to any other agent over their own. 
Unfortunately, with indivisible goods, an envy-free allocation is not guaranteed to exist: consider an example with two agents and one indivisible good.
Several natural relaxations of envy-freeness have been explored in the literature --- such as {\em envy-freeness up to one good} (EF1) \citep{Budish2011EF1, lipton2004ef1}, {\em envy-freeness up to a less desired good} (EFL) \citep{barman2018}, {\em envy-freeness up to any good} (EFX) \citep{caragiannis2019unreasonable}, and {\em envy-freeness up to $k$ hidden goods} (HEF-$k$) \citep{hosseini2020}.

An allocation is EF1 if whenever an agent envies another agent, the envy can be eliminated by removing some item from the other agent's allocated bundle. An allocation is EFL if, roughly speaking, whenever an agent envies another agent, the envy can be eliminated by removing some ``small'' (in the first agent's perspective) item from the other agent's allocated bundle. An allocation is EFX if whenever an agent envies another agent, the envy can be eliminated by removing \emph{any} item from the other agent's allocated bundle. 
Finally, an allocation is (uniformly) HEF-$k$ (or uHEF-$k$) if there are $k$ or fewer agents who can each ``hide'' a single good from their allocated bundle (so that they themselves can see it but all other agents are unaware of it) and the resulting allocation is envy-free. 
Note that any EFX allocation is EFL, and any EFL allocation is EF1. 
Furthermore, any uHEF-$k$ allocation is EF1 as well, for any $0 \leq k \leq n$, where $n$ is the number of agents. Finally, any envy-free allocation is trivially EFX as well as uHEF-$0$. To our knowledge, these two fairness notions are the strongest relaxations of envy-freeness that have been considered in the literature. The uHEF-$k$ requirement essentially interpolates between EF1 and envy-freeness by means of the parameter $k$, while EFX is just a stronger global requirement on the envy. Note that these two notions are incomparable; there are allocations that are EFX but not uHEF-$(n - 1)$, and ones that are uHEF-$1$ but not EFX (Appendix \ref{apdx:intro}).

EF1 allocations are guaranteed to exist for any instance of the fair allocation problem, and can in fact be computed in polynomial time \citep{lipton2004ef1}. EFL allocations are guaranteed to exist for additive valuations, and can also be computed in polynomial time \citep{barman2018}. On the other hand, the existence of EFX or uHEF-$k$ allocations are not known beyond some very special cases. In fact, the first of these remains one of the biggest open questions in this subfield.

We introduce a relaxation of these fairness criteria, where agents are represented by vertices on a fixed graph and allocations only need to satisfy the relaxed envy constraint for all neighboring pairs of agents in the graph. For hidden envy-freeness, this amounts to an agent needing to hide a good in order to eliminate the envy only from its neighbors in the graph. For EFX, this amounts to only needing to satisfy the EFX criterion among the pairs of agents corresponding to the graph edges. This reduces to the usual notion of a uHEF-$k$ or EFX allocation when the underlying graph is complete.

In addition to being a generalization of both these fairness constraints, this model is also quite natural, as it captures envy under partial information. In the real world, agents typically do not envy other agents whose allocated bundles they are unaware of. In these cases, it suffices to only consider pairs of agents who are aware of each other and therefore know only each other's allocated bundles.

\subsection{Our Contributions}
We study graph-based relaxations of hidden envy freeness (HEF) and envy freeness upto any good (EFX).

In Section \ref{sec:hidden}, we discuss hidden envy-freeness on graphs. Specifically, we show that for \emph{any} graph $G$ with a vertex cover of size $k$, a round-robin protocol achieves a uHEF-$k$ allocation. We also show that this is tight, in that there exists an instance of the problem on $G$ for which we cannot do better than uHEF-$k$. However, we do show graphs on which there are instances where the optimal $k$ is bounded away from the size of the minimum vertex cover.

Shifting to the EFX criterion, our main theoretical results are in Section \ref{sec:theory}, in the presence of goods as well as chores. For goods, in Section \ref{subsec:stars}, we show the existence of EFX allocations on stars, and then generalize it in two ways. We show that these allocations can be computed efficiently in the case of additive valuations. Similarly, in Section \ref{subsec:paths}, we start by showing the existence of EFX allocations on three-edge paths, and then generalize this result. In Section \ref{subsec:chores}, we show that the results from Section \ref{subsec:stars} hold for chores as well. Furthermore, in the presence of goods \emph{and} chores, we show that for lexicographic valuations, we can find EFX allocations on all graphs with diameter at least $4$, whereas they are known to not exist in general \citep{goodsandchores}. In Section \ref{sec:empirical}, we present an algorithm (Section \ref{subsec:algorithm}) that empirically works for all instances of the problem if the underlying graph is a path, generated using real-world data from Spliddit. 



%% file: arxiv_02.tex
\section{Preliminaries and Notation}\label{sec:prelims}
We have $n$ {\em agents}, $N = \{1, 2, \dots, n\}$ and $m$ {\em goods}, $M = \{g_1, g_2, \dots, g_m\}$. Each agent $i$ has a {\em valuation function} $v_i: 2^M \to \R_+ \cup \{0\}$ over the set of goods. We present results for two kinds of valuation functions. We call a valuation function {\em general} if the only constraint placed on it is monotonicity, i.e., for any $S \subseteq T \subseteq M$, $v_i(S) \le v_i(T)$. We call a valuation function {\em additive} if the value of each subset $S \subseteq M$ is the sum of the values of the goods in $S$, i.e. $v_i(S) = \sum_{g \in S} v_i(\{g\})$. We write $v_i(g)$ instead of $v_i(\{g\})$ for readability.

For agents $i, j \in N$ and goods $g_k, g_\ell \in M$, we write $g_k \succ_i g_\ell$ to mean $v_i(g_k) > v_i(g_\ell)$ (i.e. agent $i$ prefers good $g_k$ to $g_\ell$). We define $S \succ_i T$ analogously for subsets $S, T \subseteq M$. Agents $i$ and $j$ are said to have {\em identical} valuation functions iff for all $S \subseteq M$, $v_i(S) = v_j(S)$. Agents $i$ and $j$ with additive valuations are said to have {\em consistent} valuation functions iff for all $g_k, g_\ell \in M$, $g_k \succ_i g_\ell$ iff $g_k \succ_j g_\ell$ (i.e. the two agents have the same preference \emph{orders} for the goods, but not necessarily the same valuations). Identical valuations are consistent, but the converse is not necessarily true.

An {\em allocation} is a partition of the set of goods $M$ to agents $N$, represented by a tuple $X = (X_1, X_2, \dots, X_n)$ where $X_i$ is the subset of $M$ received by agent $i$. We typically refer to $X_i$ as the {\em bundle} allocated to agent $i$.
For a bundle $X_i$ and good $g$, we will write $X_i + g$ or $X_i - g$ to denote $X_i \cup \{g\}$ or $X_i \setminus \{g\}$ respectively.
Given an allocation $X$, we say an agent $i$ {\em envies} an agent $j$ if $X_j \succ_i X_i$.

Let $G = (N, E)$ be an (undirected) graph on $n$ vertices ${1, \ldots, n}$, where the set of vertices corresponds to the set of agents $N$. When $G$ is undirected, we use $(i, j)$ to denote an \emph{undirected} edge between agents $i$ and $j$. We will consider our notions of envy on these graphs.

\subsection{Hidden Envy}\label{subsec:hidden}

An allocation $X$ is said to be \emph{envy-free up to $k$ hidden goods} (HEF-$k$) if there exists a subset $S \subseteq M$ with $|S| \leq k$, such that for every pair of agents $i, j \in N$, we have $v_i(X_i) \geq v_i(X_j \setminus S)$. If in addition, we have $|S \cap X_i| \leq 1$ for all $i$, we say the allocation $X$ is \emph{envy-free up to $k$ uniformly hidden goods} (uHEF-$k$). There are instances that admit HEF-$k$ allocations but not uHEF-$k$ allocations \citep{hosseini2020}.

Observe that if an allocation is uHEF-$k$ for some $k$, then it is EF1. Furthermore, any instance with additive valuations has a uHEF-$(n - 1)$ allocation, by a round-robin protocol. Conversely, uHEF-$k$ allocations for $k < n - 1$ may not exist, e.g. if all agents have identical additive valuations, and there are $n - 1$ goods, forcing all goods to be hidden in order to appease the agent who misses out.

When the agents are arranged on the graph $G$, our aim is to output an allocation $X$ of the set of goods $M$ among the agents such that there is some subset $S \subseteq M$, such that on ``hiding'' this set $S$, the envy along the edges of $G$ disappears; formally, for every edge $(i, j) \in E$, $v_i(X_i) \geq v_i(X_j \setminus S)$. We call such an allocation a {\em $G$-HEF-$k$} allocation, where $k = |S|$. If in addition, we have $|S \cap X_i| \leq 1$ for all $i \in N$, then we call such an allocation a $G$-uHEF-$k$ allocation. Our goal is to find such an allocation that minimizes $k$. In general, we will talk interchangeably about an agent $i$ and the vertex $i$ of $G$.

We define the \emph{neighborhood} of $i$ in $G$, $\text{Nbd}_G(i)$, as the set of vertices adjacent to $i$ in $G$, $\{j \in N: (i, j) \in E\}$.

\subsection{EFX}\label{subsec:efx}

Consider an allocation $X$ where agent $i$ envies agent $j$. We say the agent $i$ {\em strongly envies} an agent $j$ if there exists some good $g \in X_j$ such that $X_j - g \succ_i X_i$. We sometimes say the {\em strong envy} in this case equals $\max \{\max_{g \in X_j}(v_i(X_j - g) - v_i(X_i)), 0\}$. An allocation without any strong envy (i.e. where the strong envy is zero for all pairs $i, j \in N$) is an {\em EFX allocation}.

As in Section \ref{subsec:hidden}, when the agents are arranged on the graph $G$, our goal is to output an allocation $X$ of the set of goods $M$ among the agents $N$ such that there is no edge $(i, j) \in E$ with agent $i$ strongly envious of agent $j$. Informally, we wish to allocate the set of goods $M$ among the agents, who correspond to vertices of $G$, but we only care about maintaining the EFX criterion along each of the edges. We call such an allocation a {\em $G$-EFX allocation}. As before, we will talk interchangeably about an agent $i$ and the vertex $i$ of $G$; we also define the neighborhood of $i$ in $G$ as before.





We make two simple observations. First, if $G$ has multiple connected components, it suffices to solve the problem for any one of those components, say $G_1$, as that same allocation is trivially EFX on all of $G$. This is because the allocation of the empty set of goods is trivially EFX.
Second, if $G$ consists of at most three vertices, then an EFX allocation certainly exists on $G$, as in fact $K_2$ and $K_3$ are known to have EFX allocations \citep{efxcanon, threeagents} (under additive valuations).

Therefore, WLOG, in Sections \ref{sec:theory} and \ref{sec:empirical}, we restrict our attention to connected graphs with $n \geq 4$ vertices. We remark here that complete, exact EFX allocations are known in special cases for these graphs: when all agents have consistent valuations, when all agents have one of two different types of valuations, when each item can take one of two possible values, when valuations are submodular with binary marginal gains, or when valuations are lexicographic~\citep{efxcanon, twotypes, amanatidis2021maximum, babaioff2021fair, goodsandchores}.

\subsection{Chores}

An instance of chore division has a set of $n$ agents $N$, and a set of $m$ chores $M = \{c_1, c_2, \dots, c_m\}$. Agents' valuation functions $v_i: 2^M \to \mathbb{R}_{-} \cup \{0\}$ map sets of chores to nonpositive values. 

Although HEF-$k$ cannot be defined for chores, we can modify the definition of EFX for chores. Given an allocation $X$ where $i$ envies $j$, we say that $i$ strongly envies $j$ if there is some chore $c \in X_i$ such that $X_j \succ_i X_i - c$. We can define the amount of strong envy as $\max \{\max_{c \in X_i}(v_i(X_j) - v_i(X_i - c)), 0\}$. 

It is also possible to define problem instances with both goods and chores.
We assume, in such instances, an item is either a good for all agents or a chore for all agents. 
Agent $i$ strongly envies agent $j$ if there is some good $g$ in $X_j$ such that $X_j - g \succ_i X_i$ or some chore $c \in X_i$ such that $X_j \succ_i X_i - c$, and the amount of strong envy is $\max \{\max_{g \in X_j}(v_i(X_j - g) - v_i(X_i)), \max_{c \in X_i}(v_i(X_j) - v_i(X_i - c)), 0\}$. When dealing with goods and chores, we usually write the items as $M = \{o_1, o_2, \dots, o_m\}$, and we denote the set of goods as $M^+ \subseteq M$ and chores as $M^- \subseteq M$.

In all cases, an allocation without any strong envy is an EFX allocation. We define additive, general, identical, and consistent valuations analogously to their definitions for goods. 

Due to space constraints, we provide proof sketches throughout, relegating all details to the Appendices.

%% file: arxiv_03.tex
\section{Theoretical Results for Hidden Envy-Freeness on Graphs}\label{sec:hidden}

Let $G = G_1 \cup \dots \cup G_u$ be an undirected graph consisting of connected components $\{G_1, \dots, G_u\}$. 
In this section, we present an algorithm which takes as input any vertex cover $C$ of any connected component $G_j$ and outputs a $G$-uHEF-$|C|$ allocation in polynomial time. The algorithm is a modification of the well-known round robin algorithm. 

In the round robin algorithm, all goods start off unallocated. The algorithm proceeds in rounds. At each round, agents are given a chance to pick a good from the set of all unallocated goods one by one. The algorithm terminates when there are no goods left. We make two changes to this algorithm to create \textsc{Vertex Cover Round Robin}. First, given a vertex cover $C$ of the connected component $G_j$, we ignore the agents not present in $G_j$ --- we pretend like they do not exist and give them an empty bundle. Second, at every round, the agents in the vertex cover $C$ pick a good first followed by the agents in $G_j \setminus C$. We present pseudocode in Algorithm \ref{alg:hidden-envy-algo}.

\begin{algorithm}[h]
    \caption{\textsc{Vertex Cover Round Robin}}
    \begin{algorithmic}
        \Require {A vertex cover $C$ of the component $G_j$}
        \Ensure {A $G$-uHEF-$|C|$ allocation $X$ with hidden goods $S$}
        \State $U \gets M$ \Comment{$U$ stores the unallocated goods}
        \State $X \gets (X_1, X_2, \dots, X_n) = (\varnothing, \varnothing, \dots, \varnothing)$
        \State $S \gets \varnothing$
        \While{$U \ne \varnothing$}
        \For{$i \in C$}
            \If{$U \ne \varnothing$}
            \State $g \gets$ some good in $\argmax_{g' \in U} v_i(g')$
            \State $X_i \gets X_i + g$
            \State $U \gets U - g$
            \If{$|X_i| = 1$}
                \State $S \gets S \cup X_i$
            \EndIf
            \EndIf
        \EndFor
        \For{$i \in G_j \setminus C$}
            \If{$U \ne \varnothing$}
            \State $g \gets$ some good in $\argmax_{g' \in U} v_i(g')$
            \State $X_i \gets X_i + g$
            \State $U \gets U - g$
            \EndIf
        \EndFor
        \EndWhile
        \Return $X, S$
    \end{algorithmic}
    \label{alg:hidden-envy-algo}
  \end{algorithm}
  
When agents have additive valuations, \textsc{Vertex Cover Round Robin} outputs a $G$-uHEF-$|C|$ allocation in polynomial time. Note that this allocation is also trivially $G$-HEF-$|C|$. 

\begin{restatable}{theorem}{thmhiddenvertexcover}\label{thm:hiddenvertexcover}
Let $G = (N, E)$ be an undirected graph with connected components $\{G_1. G_2, \dots, G_u\}$, where some connected component $G_j$ has a vertex cover $C$ of size $k$. Consider any instance of agents and goods defined on this graph. Algorithm \ref{alg:hidden-envy-algo} (\textsc{Vertex Cover Round Robin}) with input vertex cover $C$ outputs a $G$-uHEF-$k$ allocation, under additive valuation functions.
\end{restatable}
\begin{proof}[Proof Sketch]
On hiding the first $k$ items allocated in Algorithm \ref{alg:hidden-envy-algo}, each agent thinks of themselves as having been the first to pick every round in a round robin protocol. Vertices outside $C$ are independent, so all edges are accounted for.
\end{proof}

Note that this means that on graphs very large independent sets, we can find $G$-uHEF-$k$ allocations for small $k$. In particular, on a star, this means we can always find a $G$-uHEF-$1$ allocation. Note also that Algorithm \ref{alg:hidden-envy-algo} runs in polynomial time \emph{given} $C$. If we are not given $C$, then finding a vertex cover in $G$ of size $k$ is a canonical NP-complete problem.

Even though finding the minimum vertex cover is hard, it is worth noting that there is a simple 2-approximation algorithm (folklore) that can be used to compute an approximate minimum vertex cover; this approximate minimum vertex cover can then be used to compute a $G$-uHEF-$k'$ allocation (via Theorem \ref{thm:hiddenvertexcover}). For graphs with small vertex covers, this procedure outputs an allocation which hides significantly fewer goods than the previous best guarantee of $n-1$ \citep{hosseini2020}.

We now show that Theorem \ref{thm:hiddenvertexcover} is tight, in the following sense.

\begin{restatable}{theorem}{thmhiddentight}\label{thm:hiddentight}
Let $G = (N, E)$ be an undirected graph with connected components $\{G_1, G_2, \dots, G_u\}$. Let $C_1, C_2, \dots, C_u$ be minimum vertex covers for $G_1, G_2, \dots, G_u$. Let $k = \min_{t \in [u]} |C_t|$. Then, there is an instance of agents and goods defined on $G$ (with additive valuation functions) such that there is no $G$-HEF-$k'$ allocation for $k' < k$.
\end{restatable}

\begin{proof}[Proof Sketch]
We define $n^3$ goods, and give the agents identical additive valuations, where their value for each good is (roughly) $1 + \varepsilon$ for some sufficiently small $\varepsilon$. Fix an arbitrary allocation $X$. By the pigeonhole principle, some $G_i$ receives at least $n^2$ goods. If every agent in $G_i$ receives at least one good, then in order to cover each edge, we need to hide at least $|C_i|$ goods. Otherwise, there are agents $j, k \in V(G_i)$, with $|X_j| = 0$ and $|X_k| = n$. To eliminate envy along the path from $j$ to $k$, we need to hide $n$ goods.
\end{proof}

Theorem \ref{thm:hiddentight} provides a lower bound for the minimum number of hidden goods that can be guaranteed by an algorithm. It says that we cannot guarantee, for all instances, hiding fewer goods than the size of the minimum vertex cover of the graph. This shows that Theorem \ref{thm:hiddenvertexcover} is tight. 

Note that our upper bound applies to the stronger notion of uHEF, but our lower bound applies to the weaker notion of HEF. Thus, the upper bound and lower bound apply to both definitions of hidden envy --- HEF and uHEF. The tightness of our lower bound also shows that computing an allocation which minimizes the number of hidden goods is NP-hard; this applies to both HEF and uHEF.

\begin{restatable}{corollary}{corhiddenhardness}\label{cor:hidden-hardness}
Given a graphical fair allocation instance on $G$, and an integer $k$, the problems of 
\begin{inparaenum}[(a)]
    \item deciding if a $G$-HEF-$k$ allocation exists, and 
    \item deciding if a $G$-uHEF-$k$ allocation exists
\end{inparaenum}
are NP-complete. 
\end{restatable}
Finally, we remark that despite the tightness in Theorem \ref{thm:hiddentight}, there are instances of the problem where the minimum $k$ admitting a $G$-uHEF-$k$ allocation is bounded away (with an arbitrarily large gap) from the size of a minimum vertex cover. 

\begin{restatable}{proposition}{hiddennottight}\label{prop:hiddennottight}
For any $n \ge 3$, there is a graph $G = (N, E)$ and an instance of the allocation problem with $m = n$ that admits a $G$-uHEF-$2$ allocation, while every vertex cover of $G$ has size $\Theta(n)$.
\end{restatable}

%% file: arxiv_04.tex
\section{Theoretical Results for EFX on Graphs}\label{sec:theory}


In this section, we show $G$-EFX algorithms on two simple graphs – the star $K_{1, n-1}$, and the path $P_4$, generalize them to more complex classes, and analyze the problem in the presence of chores.

\subsection{The Star and its Generalizations}\label{subsec:stars}
A {\em star} consists of a ``central'' vertex and an arbitrary number of ``outer'' vertices, each with an edge only to the central vertex (e.g., $K_{1, 5}$ in Figure \ref{subfig:star}). We consider stars with $n - 1$ outer vertices, to maintain consistency with the fact that there are $n$ agents in total. We start with a warm-up problem.

\input{arxiv_star-graph-colored}

\begin{restatable}{proposition}{propgefxstar}\label{prop:gefx:star}
For all $n \geq 1$, when $G$ is the star $K_{1, n-1}$, a $G$-EFX allocation exists for agents with general valuations.
\end{restatable}

\begin{proof}[Proof sketch]
We compute an EFX allocation for $n$ identical copies of the central agent using \citep{efxcanon} to obtain $n$ bundles. The outer agents choose from these, and the last bundle goes to the center.
\end{proof}

Note, immediately, the unique benefit of the graph structure here; it allows us to compute an approximate EFX allocation for $n$ agents using a simple algorithm. 
Proposition \ref{prop:gefx:star} can be generalized to a larger class of graphs, each consisting of a central {\em group} of vertices all having the same valuation function. The remaining vertices can have arbitrary neighborhoods among the central vertices, but cannot have edges among themselves. We will refer to the central group of nodes $N' \subseteq N$ as the {\em core} vertices (or agents), and the nodes $N\setminus N'$ as the {\em outer} vertices. Note that the outer vertices induce an independent set $G[N\setminus N']$. We give a few examples of these graphs in Figure \ref{fig:star}.

\begin{restatable}{theorem}{thmresultone}\label{thm:result-1}
Suppose agents have general valuations, and $G = (N, E)$ consists of a core set of agents $N' \subseteq N$ with identical valuations, with $N\setminus N'$ inducing an independent set in $G$. Then a $G$-EFX allocation is guaranteed to exist.
\end{restatable}

\begin{proof}[Proof sketch]
We compute an allocation for $n$ identical copies of any core vertex using \citep{efxcanon}, and then have the outer agents choose from the resulting bundles before the core agents choose theirs.
\end{proof}

\begin{restatable}{corollary}{resultonecorollary}\label{corr:4-agents}
Suppose there are four agents with general valuation functions, with two of them having identical valuations. Then, an allocation exists that is EFX for all but possibly one pair of agents.
\end{restatable}

\begin{proof}[Proof sketch]
Take the identical agents as the core, and the other agents $\{u_1, u_2\}$ adjacent to both core vertices, and apply Theorem \ref{thm:result-1}. This is EFX on all pairs except possibly $\{u_1, u_2\}$.
\end{proof}

Proposition \ref{prop:gefx:star} also admits a second generalization for additive valuations which mildly relaxes the requirement that all the core vertices $N'$ have the same valuation function, but places an additional restriction on the outer vertices. Now, core vertices only need to have \emph{consistent} valuation functions, i.e., the same (weak) ranking over the set $M$ of goods. However, each outer vertex can only have a neighborhood in $N'$ with identical valuation functions.

\begin{restatable}{theorem}{thmresulttwo}\label{thm:result-2}
Under additive valuations, let $G = (N, E)$ consist of a core set of agents $N' \subseteq N$ with consistent valuations. Let $N' = N'_1 \sqcup \ldots \sqcup N'_K$, where all agents in $N'_k$ have the same valuation function $v_k$. Suppose every $i \in N \setminus N'$ has its neighborhood $\mathrm{Nbd}_G(i)\subseteq N'_k$ for some $k$. Then a $G$-EFX allocation is guaranteed to exist.
\end{restatable}

\begin{proof}[Proof sketch]
We imagine each outer agent has the same valuation as its core neighbors, and compute an allocation on that instance using \citep{efxcanon}. Each outer agent then chooses a bundle (from the relevant pool) before their core neighbors do.
\end{proof}

Theorem \ref{thm:result-2} implies Theorem \ref{thm:result-1} when agents have additive valuations, simply by taking $K = 1$. While Theorem \ref{thm:result-2} applies to several interesting and natural graphs (see Figures \ref{subfig:star-general-2} and \ref{subfig:star-general-3}, e.g.), a limiting constraint is that each vertex in $N\setminus N'$ can only have identical neighbors in $N'$. One might be tempted to simultaneously generalize Theorems \ref{thm:result-1} and \ref{thm:result-2} (in the additive case) and try to prove the existence of $G$-EFX allocations in graphs where the outer agents have arbitrary neighborhoods among the consistent core agents. The problem with this approach is one of ambiguity between which core agents an outer agent should ``act'' like in order to use the result in \citep{efxcanon}. The following example illustrates this difficulty.

\input{arxiv_counterexample}

\begin{example}\label{ex:counterexample}
Consider the allocation instance with three agents $\{1, 2, 3\}$ and six goods $\{g_1, \dots, g_6\}$ defined over $P_3$ (see Figure \ref{fig:counterexample-graph}). Suppose the agents have additive valuations, given in Figure \ref{tab:counterexample-valuations}. 

Agents $1$ and $3$ have (weakly) consistent valuations, but they have a common neighbor $2$ whose valuation is inconsistent with them. If we try to follow the same proof as Theorem \ref{thm:result-2}, we would first compute an EFX allocation for consistent valuations and then reallocate the bundles by letting $2$ pick first. However, we have to first choose a valuation profile to give agent $2$. If we give agent $2$ the same valuation as agent $3$, we can easily find an EFX allocation for consistent valuations $(v_1, v_3, v_3)$: e.g. $Y = (Y_1, Y_2, Y_3) = (\{g_1, g_2\}, \{g_3, g_4\}, \{g_5, g_6\})$. If we now give agent $2$ their highest-valued bundle between $Y_2$ and $Y_3$, we get the final allocation $X = (\{g_1, g_2\}, \{g_5, g_6\}, \{g_3, g_4\})$. However this is not $G$-EFX since agent $2$ strongly envies agent $1$.

The problem arose when agent $2$ was asked to pick between $Y_2$ and $Y_3$, and not allowed to choose agent $Y_1$. If we modified the algorithm and let agent $2$ pick $Y_1$, no matter how we distribute the remaining bundles, agent $1$ would always strongly envy agent $2$.
\end{example}

We remark that the instance in Example \ref{ex:counterexample} admits a $G$-EFX allocation; Proposition \ref{prop:gefx:star} implies this trivially. The obstacle, therefore, arises from the choices made in our algorithm.

An important question when discussing EFX allocations is that of computational efficiency. This is known to be an intractable problem in general; \citet{efxcanon} show that computing EFX allocations even when the valuations are identical has exponential query complexity. 
However, we show that for additive valuations, the $G$-EFX allocations from Theorem \ref{thm:result-1} and \ref{thm:result-2} can be computed in polynomial time. 

\begin{restatable}{theorem}{thmtimecomplexityone}\label{thm:time-complexity-1}
When agents have additive valuations, the $G$-EFX allocations in Theorems \ref{thm:result-1} and \ref{thm:result-2} can be computed in $O(m n^3)$ time, where $n$ and $m$ are the number of agents and number of goods respectively.
\end{restatable}

\subsection{The Three-Edge Path and its Generalization}\label{subsec:paths}
We now move to another simple graph that is the starting point for our next set of results: the three-edge path graph $P_4$.

\input{arxiv_three-edge-path-colored}

\begin{restatable}{proposition}{threeedgepath}\label{prop:three-edge-path}
When $G$ is the three-edge path $P_4$, a $G$-EFX allocation exists for agents with arbitrary general valuations.
\end{restatable}

\begin{proof}[Proof sketch]
Suppose the agents are numbered $1, 2, 3, 4$ in order. We use \citep{twotypes} to compute an allocation $X = (X_1, X_2, X_3, X_4)$ for valuation functions $(v_2, v_2, v_3, v_3)$, and have $1$ choose from $\{X_1, X_2\}$ before giving the other one to $2$ (and similarly for $3$ and $4$).
\end{proof}

Exactly as in Section \ref{subsec:stars}, this result can be generalized to a larger class of graphs. Again, we have \emph{core vertices} $N' \subseteq N$, corresponding to agents having one of two distinct valuation functions, say $v_k$ and $v_\ell$. The outer agents (in $N\setminus N'$) induce an independent set as before, and furthermore, they can have arbitrary neighborhoods among any one of the two types of agents in $N'$ (e.g., Figure \ref{fig:three-edge-path}). Formally, $N'$ can be partitioned as $N' = N'_k \sqcup N'_\ell$, with all agents in $N'_r$ with valuation $v_r$, for $r \in \{k, \ell\}$. The outer vertices in $N\setminus N'$ have arbitrary valuation functions, and for each $i \in N\setminus N'$, we have $\mathrm{Nbd}_G(i) \subseteq N'_k$ or $\mathrm{Nbd}_G(i) \subseteq N'_\ell$.

When $G$ is of the form above, we can find an EFX allocation.

\begin{restatable}{theorem}{thmresultthree}\label{thm:result-3}
Suppose agents have general valuations, and $G = (N, E)$ is of the form described above, i.e., consists of a core set of vertices $N' \subseteq N$ with two types of valuations, and all remaining agents in $N' \subseteq N$ with arbitrary valuations, but neighborhoods restricted to any of the two core groups of agents. Then a $G$-EFX allocation is guaranteed to exist.
\end{restatable}

\begin{proof}[Proof sketch]
As before we use \citep{twotypes} to compute an allocation on two types of goods, imagining outer vertices to share the valuation profiles of their neighbors. Each outer agent then chooses a bundle (from the relevant pool) before their core neighbors do.
\end{proof}

We conclude this subsection by observing that Theorem \ref{thm:result-3} immediately shows the existence of $G$-EFX allocations for several classes of small graphs (Figure \ref{fig:three-edge-path}). Notably among these are the graph consisting of two arbitrary stars connected at their central vertices, and the {\em four-edge path} $P_5$ where any two of the degree-$2$ vertices have the same valuation profile.

\subsection{Chores}\label{subsec:chores}

In this subsection, we will investigate the problem in the presence of chores. To begin, assume that there are no goods. It is known that an EFX allocation exists for chores when agents have consistent additive valuations \citep{li2022efxchores}. Moreover, this allocation can be computed in polynomial time.
Using this result, we show that analogous results to Theorems \ref{thm:result-1} and \ref{thm:result-2} hold in this case. 
The proofs are very similar to that of Theorems \ref{thm:result-1} and \ref{thm:result-2} and are therefore omitted. 


\begin{restatable}{theorem}{thmresultonechores}\label{thm:result-1chores}
Suppose we only have chores, under additive valuations. Let $G = (N, E)$ consist of a core set of agents $N' \subseteq N$ with identical valuations, with $N\setminus N'$ inducing an independent set in $G$. Then a $G$-EFX allocation is guaranteed to exist and it can be computed in polynomial time.
\end{restatable}

\begin{restatable}{theorem}{thmresulttwochores}\label{thm:result-2chores}
Suppose we only have chores, under additive valuations. Let $G = (N, E)$ consist of a core set of agents $N' \subseteq N$ with consistent valuations. Let $N' = N'_1 \sqcup \ldots N'_K$, where all agents in $N'_k$ have the same valuation function $v_k$, and every vertex $i \in N \setminus N'$ has its neighborhood $\mathrm{Nbd}_G(i)\subseteq N'_k$ for some $k$. Then a $G$-EFX allocation is guaranteed to exist and it can be computed in polynomial time.
\end{restatable}



We remark that an analogous result to Theorem \ref{thm:result-3} is hard to obtain for chores, as EFX allocations are not known to exist for chores under two types of valuations, in the style of \citet{twotypes}.

To conclude this section, consider the setting with \emph{both} goods and chores, where it is known that an EFX allocation is not guaranteed to exist for general valuations \citep{mixedgeneralcounter}. Recently, it was shown that they do not exist even when the valuations are \emph{lexicographic} \citep{goodsandchores}.

\begin{definition}
A valuation function $v_i$ is \emph{lexicographic} if $M$ has a priority order $o_1 \succ_i \ldots \succ_i o_m$, such that if $o_1 \in M^+$ (resp. $M^-$), $i$ prefers any bundle with (resp. without) $o_1$ over any other bundle; subject to that, agent $i$ prefers any bundle with (resp. without) $o_2$ over any other bundle, and so on.
\end{definition}

The \emph{diameter} of a graph $G$ is the maximum length of a shortest path in $G$, $\max_{i, j \in V}\mathrm{dist}_G(i, j)$. We show that if $G$ has sufficiently long diameter, it admits an EFX allocation under lexicographic valuations on goods and chores. 

\begin{restatable}{proposition}{lexicographic}\label{prop:lexicographic}
Any graph $G$ with diameter $d \geq 4$ admits a $G$-EFX allocation under lexicographic valuations on goods and chores, and this allocation can be computed in polynomial time.
\end{restatable}

\begin{proof}[Proof sketch]
For agents $u$ and $v$ at distance at least $4$, all neighbors of $u$ receive one good, all neighbors of $v$ receive one chore, and the remaining goods and chores are given to $u$ and $v$ respectively. All other agents receive empty allocations. 
\end{proof}

The class of graphs with diameter at least $4$ is large, and it includes most sparse graphs (including most trees) and graphs with large treewidth. The \emph{six degrees of separation} thought experiment and its theoretical generalizations suggest that all real-life social networks fall in this category of graphs \citep{sixdegrees}. Proposition \ref{prop:lexicographic} demonstrates that restricting EFX to graphs can sometimes enable fairness properties that are not possible in general. In particular, the result of \citet{goodsandchores} closed off EFX under additive valuations for goods and chores, but it does not hold on non-complete graphs. Our result indicates that $G$-EFX may be possible for goods and chores on many natural graphs under arbitrary additive valuations.


%% file: arxiv_star-graph-colored.tex
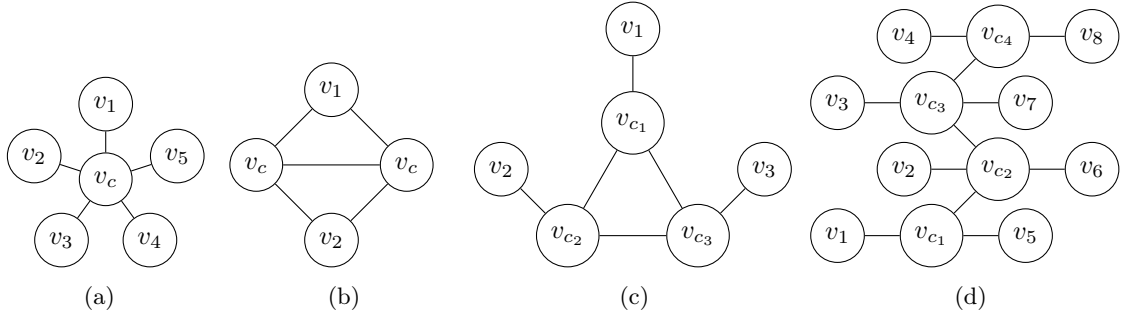
\begin{figure*}
    \centering
    \begin{subfigure}[b]{0.20\textwidth} 
    \hspace*{5pt}
    \begin{tikzpicture}[node distance = 1.25cm]
    \node[circle, draw = black] at (360:0mm) (center) {$v_c$};
	\foreach \n in {1,...,5}{
        \node[circle, draw = black] at ({18+\n*360/5}:1cm) (n\n) {$v_{\n}$};
        \draw (center)--(n\n);
    }
    \end{tikzpicture}
    \caption{}
    \label{subfig:star}
    \end{subfigure}
    \hfill
    \begin{subfigure}[b]{0.20\textwidth} 
    \begin{tikzpicture}[node distance = 1.25cm]
	\node[circle, draw = black] at ({360/4}:1cm) (v1) {$v_{1}$};
	\node[circle, draw = black] at ({180}:1cm) (vc1) {$v_{c}$};
	\node[circle, draw = black] at ({360}:1cm) (vc2) {$v_{c}$};
	\node[circle, draw = black] at ({270}:1cm) (v2) {$v_{2}$};
	\draw (vc1)--(vc2);
	\draw (vc1)--(v1);
	\draw (vc1)--(v2);
	\draw (vc2)--(v1);
	\draw (vc2)--(v2);
    \end{tikzpicture}
    \caption{}
    \label{subfig:star-general-1}
    \end{subfigure}
    \hfill
    \begin{subfigure}[b]{0.28\textwidth} 
    \begin{tikzpicture}[node distance = 1.25cm]
	\node[circle, draw = black] at ({90}:1cm) (vc1) {$v_{c_1}$};
	\node[circle, draw = black] at ({210}:1cm) (vc2) {$v_{c_2}$};
	\node[circle, draw = black] at ({330}:1cm) (vc3) {$v_{c_3}$};
	\node[circle, draw = black, above of = vc1] (v1) {$v_{1}$};
	\node[circle, draw = black, above left of = vc2] (v2) {$v_{2}$};
	\node[circle, draw = black, above right of = vc3] (v3) {$v_{3}$};
	\draw (vc1)--(vc2);
	\draw (vc3)--(vc2);
	\draw (vc3)--(vc1);
	\draw (vc1)--(v1);
	\draw (vc2)--(v2);
	\draw (vc3)--(v3);
    \end{tikzpicture}
    \caption{}
    \label{subfig:star-general-2}
    \end{subfigure}
    \hfill
    \begin{subfigure}[b]{0.28\textwidth}
    \begin{tikzpicture}[node distance = 1.25cm]
	\node[circle, draw = black] (vc1) {$v_{c_1}$};
	\node[circle, draw = black, above right of = vc1] (vc2) {$v_{c_2}$};
	\node[circle, draw = black, above left of = vc2] (vc3) {$v_{c_3}$};
	\node[circle, draw = black, above right of = vc3] (vc4) {$v_{c_4}$};
	\node[circle, draw = black, left of = vc1] (v1) {$v_{1}$};
	\node[circle, draw = black, left of = vc2] (v2) {$v_{2}$};
	\node[circle, draw = black, left of = vc3] (v3) {$v_{3}$};
	\node[circle, draw = black, left of = vc4] (v4) {$v_{4}$};
	\node[circle, draw = black, right of = vc1] (v5) {$v_{5}$};
	\node[circle, draw = black, right of = vc2] (v6) {$v_{6}$};
	\node[circle, draw = black, right of = vc3] (v7) {$v_{7}$};
	\node[circle, draw = black, right of = vc4] (v8) {$v_{8}$};
	\draw (vc1)--(vc2);
	\draw (vc3)--(vc2);
	\draw (vc3)--(vc4);
	\draw (vc1)--(v1);
	\draw (vc2)--(v2);
	\draw (vc3)--(v3);
	\draw (vc4)--(v4);
	\draw (vc1)--(v5);
	\draw (vc2)--(v6);
	\draw (vc3)--(v7);
	\draw (vc4)--(v8);
    \end{tikzpicture}
    \caption{}
    \label{subfig:star-general-3}
    \end{subfigure}
    \caption{An example of a star graph and its generalizations. Each node $i$ in any of the graphs above is labeled by the valuation function $v_i$ of the corresponding agent. The valuation functions $v_{c_1}, v_{c_2}, v_{c_3}$ and $v_{c_4}$ are consistent.}
    \label{fig:star}
\end{figure*}

%% file: arxiv_counterexample.tex
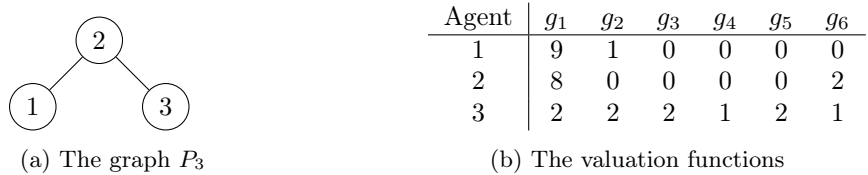
\begin{figure}
    \centering
    \begin{minipage}[b]{\linewidth}
    \begin{subfigure}[b]{0.45\textwidth}
        \hspace*{55pt}
        \begin{tikzpicture}[node distance = 1.25cm]
            \node[circle, draw = black] (x){$1$};
            \node[circle, draw = black, above right of = x]  (y){$2$};
            \node[circle, draw = black, below right of = y]  (z){$3$};
            \draw (x)--(y);
            \draw (y)--(z);
        \end{tikzpicture}
        \caption{The graph $P_3$}
        \label{fig:counterexample-graph}
    \end{subfigure} 
    \begin{subfigure}[b]{0.45\textwidth}
        \hspace*{15pt}
        \begin{tabular}{c | c c c c c c }
             Agent & $g_1$ & $g_2$ & $g_3$ & $g_4$ & $g_5$ & $g_6$\\
             \hline
             $1$ & 9 & 1 & 0 & 0 & 0 & 0 \\ 
             $2$ & 8 & 0 & 0 & 0 & 0 & 2\\ 
             $3$ & 2 & 2 & 2 & 1 & 2 & 1\\
        \end{tabular}
        \caption{The valuation functions}
        \label{tab:counterexample-valuations}
    \end{subfigure} 
    \end{minipage}
    \caption{The allocation instance in Example \ref{ex:counterexample}}
    \label{fig:counterexample}
\end{figure}

%% file: arxiv_three-edge-path-colored.tex
\begin{figure*}
    \centering
    \begin{tikzpicture}[node distance = 1.25cm]
    \node[circle, draw = black] (v1) {$v_1$};
	\node[circle, draw = black, below right of = v1] (v2) {$v_2$};
	\node[circle, draw = black, above right of = v2] (v3) {$v_3$};
	\node[circle, draw = black, below right of = v3] (v4) {$v_4$};
	\draw (v1)--(v2);
	\draw (v2)--(v3);
	\draw (v3)--(v4);
    \end{tikzpicture}
    \qquad
    \begin{tikzpicture}[node distance = 1.25cm]
	\node[circle, draw = black] (v1) {$v_{1}$};
	\node[circle, draw = black, above right of = v1] (v2) {$v_{2}$};
	\node[circle, draw = black, below right of = v2] (v3) {$v_{2}$};
	\node[circle, draw = black, above right of = v3] (v4) {$v_{3}$};
	\node[circle, draw = black, below right of = v4] (v5) {$v_{4}$};
	\draw (v1)--(v2);
	\draw (v3)--(v2);
	\draw (v3)--(v4);
	\draw (v4)--(v5);
    \end{tikzpicture}
    \\
    \vspace{0.5cm}
    \begin{tikzpicture}[node distance = 1.25cm]
    \node[circle, draw = black] (v1) {$v_1$};
	\node[circle, draw = black, above right of = v1] (v2) {$v_2$};
	\node[circle, draw = black, right of = v2] (v3) {$v_3$};
	\node[circle, draw = black, below right of = v3] (v4) {$v_4$};
	\node[circle, draw = black, above left of = v2] (v5) {$v_5$};
	\node[circle, draw = black, above right of = v3] (v6) {$v_6$};
	\draw (v1)--(v2);
	\draw (v5)--(v2);
	\draw (v2)--(v3);
	\draw (v3)--(v4);
	\draw (v3)--(v6);
    \end{tikzpicture}
    \qquad
    \begin{tikzpicture}[node distance = 1.25cm]
    \node[circle, draw = black] (v1) {$v_1$};
	\node[circle, draw = black, right of = v1] (v2) {$v_2$};
	\node[circle, draw = black, below of = v1] (v12) {$v_1$};
	\node[circle, draw = black, right of = v12] (v22) {$v_2$};
	\node[circle, draw = black, left of = v1] (v3) {$v_3$};
	\node[circle, draw = black, left of = v12] (v4) {$v_4$};
	\node[circle, draw = black, right of = v2] (v5) {$v_5$};
	\node[circle, draw = black, right of = v22] (v6) {$v_6$};
	\draw (v1)--(v2);
	\draw (v12)--(v2);
	\draw (v1)--(v22);
	\draw (v12)--(v22);
	\draw (v1)--(v12);
	\draw (v2)--(v22);
	\draw (v3)--(v1);
	\draw (v3)--(v12);
	\draw (v4)--(v1);
	\draw (v4)--(v12);
	\draw (v5)--(v2);
	\draw (v5)--(v22);
	\draw (v6)--(v2);
	\draw (v6)--(v22);
    \end{tikzpicture}
    
    \caption{A three-edge path and its generalizations. Each node $i$ in any of the graphs above is labeled by the valuation function $v_i$ of the corresponding agent.}
    \label{fig:three-edge-path}.
\end{figure*}
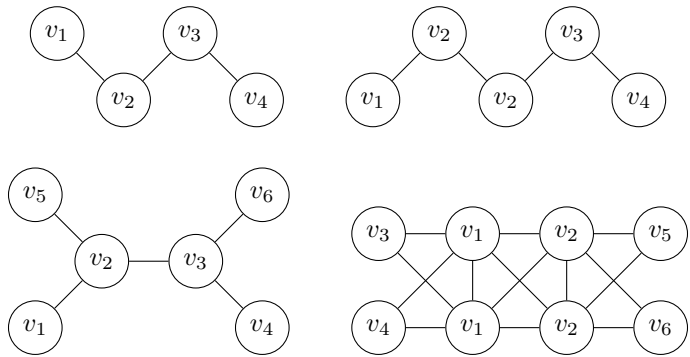

%% file: arxiv_05.tex
\section{Empirical Results for EFX on Graphs}\label{sec:empirical}

In this section, we discuss a simple algorithm that handles allocations on a more robust class of graphs than in Section \ref{sec:theory}. We do not prove that this algorithm always terminates successfully, though it is obvious by design that {\em if} it terminates, it does so with a $G$-EFX allocation. We present empirical results using our algorithm in many (representative) real instances, and show that the algorithm terminates with a $G$-EFX allocation in all of them. The typical proof technique in these contexts is to use a \emph{potential function}, i.e., a loop variant bounded below that monotonically decreases during each round, indicating progress being made~\citep{benabbou2020finding, threeagents, twotypes}. We empirically investigate different natural candidates for such a function.

\subsection{The ``Sweeping'' Algorithm}\label{subsec:algorithm}

For simplicity, all our empirical results will use the path graph $P_n$ for the underlying graph $G$. We will assume that the agents are $1, \ldots, n$ in that order along the path, and that agents have additive valuations.
The pseudocode of the algorithm can be found in Algorithm~\ref{alg:sweeping_algo}. For each edge, the algorithm reallocates the goods along that edge using Algorithm~\ref{alg:local_efx_algo}, which combines the ideas of Algorithms 4.2 and 6.1 from~\citet{efxcanon}.

The algorithm proceeds in rounds after assigning all goods in $M$ initially to agent $1$. Each round of the algorithm consists of ``sweeping'' the edges of $P_n$ step by step, from left to right in order (the {\em forward sweep}), and then back again (the {\em reverse sweep}). 

\begin{algorithm}[h]
    \caption{\textsc{Local EFX} (\citep{efxcanon})}
    \begin{algorithmic}
        \Require {Bundles $X_i, X_{i+1}$, additive valuation functions $v_i, v_{i+1}$}
        \State $M' \gets \text{Sorted}(X_i \cup X_{i+1})$ \Comment{In decreasing order of $v_i$}
        \State $(Y_1, Y_2) \gets (\varnothing, \varnothing)$
        \For{$g \in M'$}
            \State $Y' \gets \argmin_{j \in \{1, 2\}} v_i(Y_j)$
            \State $Y' \gets Y' + g$
        \EndFor
        
        \State $X_{i+1} \gets \argmax_{j \in \{1,2\}} v_{i+1}(Y_j)$
        \State $X_i \gets M' \setminus X_{i+1}$
        
    \State \Return $X_i, X_{i+1}$
    \end{algorithmic}
    \label{alg:local_efx_algo}
  \end{algorithm}
  
\begin{algorithm}[t]
    \caption{\textsc{Sweeping Algorithm}}
    \begin{algorithmic}
        \Require {$n$ agents $N$, $m$ goods $M$, additive valuation functions $v_i: 2^M \to \mathbb{R}$ for agents $i \in N$, path graph $P_n$}
        \State Initialize allocation $X = (X_1, \dots X_n)$ with $X_1 = M$ and $X_i = \varnothing$ for all $i \neq 1$
        
        \While{There is an edge $(i, i+1)$ in $P_n$ such that $i$ strongly envies $i+1$ or $i+1$ strongly envies $i$}
            \State $X^{\text{init}} \gets X$
            \For{$i \in \{1, \dots, n-1\}$}
                \State $(X_i, X_{i+1}) \gets {\textsc{LocalEFX}}(X_i, X_{i+1}, v_i, v_{i+1})$
            \EndFor
            \For{$i \in \{n-2, \dots, 1\}$}
                \State $(X_i, X_{i+1}) \gets {\textsc{LocalEFX}}(X_i, X_{i+1}, v_i, v_{i+1})$
            \EndFor
            \If{$X^{\text{init}} = X$ and some edge $(i, i+1)$ has strong envy}
                \State \Return Failure
            \EndIf
        \EndWhile
    \State \Return $X$
    \end{algorithmic}
    \label{alg:sweeping_algo}
  \end{algorithm}

This implies that on a particular step of any round, there is a well-defined edge of $P_n$ we are looking at, say the edge $(i, i + 1)$, with current bundles $(X_i, X_{i+1})$. This step of the algorithm consists of ``fixing'' the allocation on the edge $(i, i + 1)$ by re-allocating the goods in $X_i\cup X_{i+1}$ between the agents $i$ and $i + 1$ so that neither of them strongly envies the other. Recall that this is straightforward to do, as EFX allocations are known to exist for general valuations on two agents, due to \citet{efxcanon}, who use a cut-and-choose protocol, with one agent being the ``cutter'' and the other the ``chooser''\footnote{ Algorithm~\ref{alg:local_efx_algo}, while efficient, only works for additive valuations. \citet{efxcanon} give a potentially inefficient algorithm for 2 agents with general valuations, so the sweeping algorithm could still theoretically be applied with general valuations.}. Fixing a particular edge $(i, i + 1)$ might negate the EFX criterion on a previously fixed edge, such as $(i - 1, i)$; we then must fix that edge in a subsequent step, sweep or round. If an edge under consideration already meets the EFX criterion, then we do nothing and move on to the next step of the sweep. At the end of a round (i.e. we are on the edge $(1, 2)$ at the end of a reverse sweep), we check if our current allocation is $G$-EFX. If it is, we are done; otherwise, we start the next round with the current allocation.

We conclude this subsection by remarking that this algorithm can be defined for other graphs, as well as with different initialization conditions. For instance, if $G$ is a tree, we could use an in-order traversal of its edges as our ``sweeping'' order. We could also start with a random allocation of the goods among the vertices of $G$ instead of assigning all of them to one particular vertex. We relegate the analysis of these variants to future work, and for the rest of this section, only discuss path graphs with the initialization condition as stated above. In addition, because Algorithm 4.2 of \citet{efxcanon} applies to two agents with general valuations, this algorithm could be applied to instances with general valuations. For simplicity, our experiments use additive valuation functions.

\begin{table*}
    \centering
    \begin{tabular}{|c|c|c|}
    \hline
        Potential Function & Closed-Form Expression & Expectation \\
        \hline
         Total Envy $(\phi_1)$ & $\sum_{i = 1}^{n-1} \envy_v^X(i, i+1) + \envy_v^X(i+1, i)$ &  Decreasing \\ [0.1cm] 
         Total Strong Envy $(\phi_2)$ & $\sum_{i = 1}^{n-1} \efxenvy_v^X(i, i+1) + \efxenvy_v^X(i+1, i)$ & Decreasing \\ [0.1cm]
         Minimum Valuation $(\phi_3)$ &
         $\min_i v_i(X_i)$ & Increasing \\
         \hline
    \end{tabular}
    \caption{Summary of the potential functions. $\envy_v^X(i, j)$ denotes the amount by which $i$ envies $j$ in $X$. $\efxenvy_v^X(i, j)$ denotes the amount by which $i$ strongly envies $j$ in $X$ i.e. the total envy $i$ has for $j$ after dropping $j$'s worst good from $i$'s perspective.  
    }
    \label{tab:pot-functions}
\end{table*}

\subsection{Performance and Potential Functions}\label{subsec:performance}

We use the data from Spliddit \citep{spliddit}\footnote{See \href{http://www.spliddit.org/}{http://www.spliddit.org/}} for our experiments. Spliddit users set up allocation problems with any number of goods and agents. Each agent is given 1,000 points to allocate across all of the goods, in integer amounts. We assume all goods are indivisible. There are 3,392 problem instances with three or more agents. For each problem instance on $n$ agents, we number them $1$ through $n$ and create $P_n$. Code for these experiments is available on GitHub\footnote{See \href{https://github.com/justinpayan/graph_efx}{https://github.com/justinpayan/graph\_efx}}.

We found that Algorithm \ref{alg:sweeping_algo} successfully terminates with a $G$-EFX allocation in every single instance. In addition, it seems to converge extremely fast, completing all instances in a matter of seconds on an Intel Core i7 8th generation processor. The vast majority (3,087) of instances finished in only a single round, 296 required two rounds, 8 required three, and only a single instance required four rounds.

\begin{figure*}
    \centering
    \begin{minipage}[b]{\textwidth}
    \begin{subfigure}[b]{0.32\linewidth}
        \includegraphics[width=\linewidth]{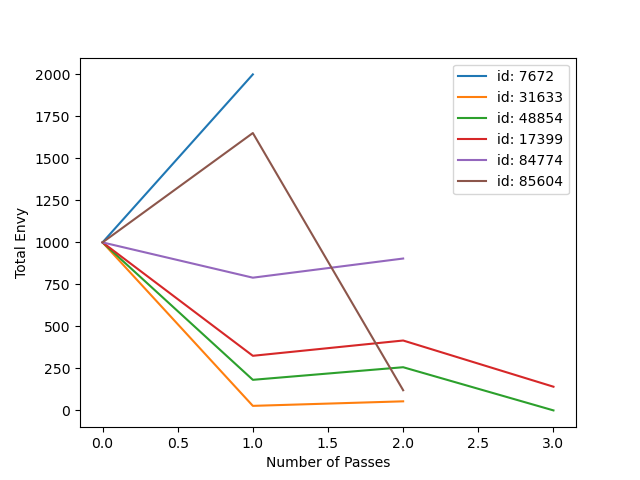}
        \caption{Total Envy}
        \label{fig:envy_over_passes}
    \end{subfigure} 
    \hfill
    \begin{subfigure}[b]{0.32\textwidth}
        \includegraphics[width=\linewidth]{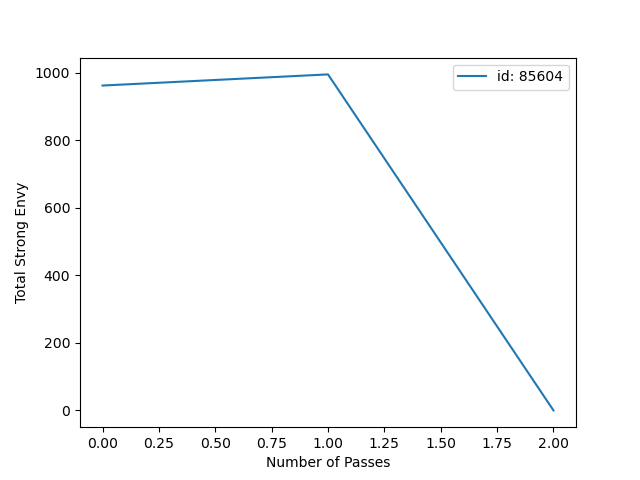}
        \caption{Total Strong Envy}
        \label{fig:strong_envy_over_passes}
    \end{subfigure} 
    \hfill
    \begin{subfigure}[b]{0.32\textwidth}
        \includegraphics[width=\linewidth]{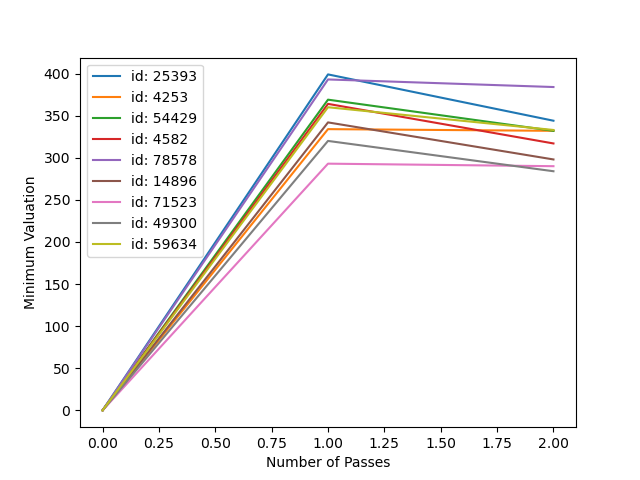}
        \caption{Minimum Valuation}
        \label{fig:min_over_passes}
    \end{subfigure} 
    \end{minipage}
    \caption{Plots of three candidate potential functions on selected instances. Neither total envy nor total strong envy are monotonically non-decreasing, and the minimum valuation is not monotonically non-increasing.}
    \label{fig:potential_fns}
\end{figure*}

In order to analyze this algorithm, it suffices to show that it terminates, as the termination condition corresponds to an EFX allocation. The standard technique for showing such a round-based algorithm terminates is to use a {\em potential function} $\phi$, a variable quantity that changes monotonically across rounds, indicating that the algorithm is making progress. Typically, the potential function is bounded by a theoretical maximum or minimum, which implies that it cannot continue to increase or decrease indefinitely. Several papers related to EFX allocations use potential functions \cite{benabbou2020finding, threeagents, twotypes}. In our case, we restrict our attention to potential function values in between \emph{rounds} rather than \emph{steps}, without loss of generality.

We investigate three reasonable candidates for potential functions, summarized in Table~\ref{tab:pot-functions}: $\phi_1$, the sum of the envy across all adjacent agents (in both directions); $\phi_2$, the sum of the \emph{strong} envy over all adjacent agents; and $\phi_3$, the minimum value realized by any agent on its assigned bundle. We expect $\phi_1$ or $\phi_2$ to be monotonically non-increasing over a run of our algorithm, as we repeatedly remove strong envy between pairs of agents. We also expect $\phi_3$ to be monotonically non-decreasing, since at each step envious agents with relatively worse bundles receive better goods. Figure~\ref{fig:potential_fns} shows the trajectories of these potential functions. Plots \ref{fig:strong_envy_over_passes} and \ref{fig:min_over_passes} show all examples with violations of monotonicity for $\phi_2$ and $\phi_3$, but there are 273 instances which violate monotonicity of $\phi_1$. While the instances violating monotonicity of $\phi_3$ violate it by a small amount, the violations for $\phi_1$ have no discernible pattern.

\begin{table}
    \centering
    
        \begin{tabular}{c | c c c c c c }
             Agent & $g_1$ & $g_2$ & $g_3$ & $g_4$ & $g_5$ & $g_6$\\
             \hline
             $1$ & 120 & 200 & 80 & 120 & 400 & 80 \\ 
             $2$ & 39 & 39 & 38 & 77 & 769 & 38\\ 
             $3$ & 994 & 1 & 1 & 1 & 2 & 1\\
        \end{tabular}
   \caption{Instance where $\phi_2$ is not non-increasing.}
    \label{tab:strong_envy_not_monotone}
\end{table}

Because there is only one instance violating monotonicity of strong envy, we present it here (Table~\ref{tab:strong_envy_not_monotone}). We remark that when looking at the edge $(i, i+1)$, our implementation always designates $i$ as the cutter and $i + 1$ as the chooser. However, if we had reversed these roles for $2$ and $3$ in the first round, we would immediately obtain a $G$-EFX allocation. This means that for all Spliddit instances, there is some sequence of cutters and choosers in our protocol that causes the overall strong envy to be monotonically non-increasing.




%% file: arxiv_06.tex
\section{Previous Work}\label{sec:previous}


The existence of EF1 allocations in general using an efficient algorithm (cycle elimination in the envy graph) was shown by \citet{lipton2004ef1}. EFL allocations were introduced by \citet{barman2018}, who also showed their existence under additive valuations.

Hidden envy was studied by \citet{hosseini2020}, who showed most known results in the realm, together with lower bounds.

EFX allocations are known to exist in certain special cases.
\citet{efxcanon} present a number of results for EFX allocations, most notably EFX for $n$ agents with identical general valuations, $n$ agents with consistent additive valuations, and for two agents with arbitrary general valuations. It is also known that EFX allocations exist for three agents with additive valuations~\citep{threeagents} and when all agents have one of two (general) valuation functions \citep{twotypes}. EFX allocations also exist for submodular valuations with binary marginal gains \citep{babaioff2021fair}. When agents can only have two possible values for all goods, the solution maximizing the product of agents' valuations (the Nash welfare) is EFX \citep{amanatidis2021maximum}.


There have been a few variants of this problem which are centered on approximately EFX allocations. \citet{efxcanon} define a $c$-EFX allocation as one in which for any pair of agents, either can drop any good from their allocated bundle, and be within a multiplicative factor of $c$ from the other agent's valuation for their bundle. They demonstrate an algorithm to achieve a $\frac{1}{2}$-EFX allocation. \citet{amanatidis2020multiple} subsequently show a single algorithm satisfying $(\phi-1)$-EFX ($\phi \approx 1.618$ is the golden ratio) along with three other fairness properties out of the scope of the current work.
\citet{amanatidis2021maximum} also present an alternative approximation rule for EFX allocations, and show that the maximum Nash welfare solution is a $\frac{1}{2}$-approximation to EFX under this new definition.

Another line of work has sought EFX allocations on a subset of the goods, leaving the rest unallocated. \citet{caragiannis2019envy} aim for high Nash welfare and EFX, showing an allocation on a subset of goods that is EFX and achieves at least half of the maximum possible Nash welfare. \citet{chaudhury2021little} give an algorithm such that no more than $n$ goods go unallocated. \citet{chaudhury2021improving} create an EFX allocation and bound the number of goods that remain unallocated using a graph theoretic function on the natural numbers called the {\em rainbow cycle number}. \citet{berger2021almost} improve these results for four agents by constructing an EFX allocation that leaves at most one good unallocated in that setting.



A few other works have modeled fairness concepts using graphs. One line of work specifies a graph structure over goods where bundles correspond to connected subgraphs~\citep{bouveret2017fair, bilo2022almost, igarashi2019pareto, bei2021price, bei2021dividing, bei2020cake, tucker2021thou}. Other works have considered graph structures over agents. \citet{beynier2019local} investigate envy-free housing allocation over a graph, where agents receive one good each and must not envy their neighbors. \citet{bredereck2022envy} likewise consider envy-free allocations over a graph, with the goal of determining in polynomial or parametrized polynomial time if a envy-free allocation exists on the graph. \citet{aziz2018knowledge} assume that agents can only view the allocations of adjacent agents in a graph. They seek allocations that are \emph{epistemically} envy-free, where no vertex envies its neighbors, and furthermore, for any vertex $x$, there is an allocation of the remaining goods (other than the ones allocated to $x$ and its neighbors) to the other agents so that $x$ does not envy any other agent. \citet{epistemicEFX} have considered an \emph{epistemic EFX} variant. A few papers study scenarios which explicitly limit the operations of distributed algorithms to a graph structure, and aim to satisfy envy-freeness or maximize the minimum agent value \cite{beynier2018fairness, lange2019optimizing, eiben2020parameterized, kaczmarczyk2019fair, varricchione2021complexity}. 
Although these works have similar motivation to ours, none of them to our knowledge address the question of finding EFX or HEF allocations on a graph.

%% file: arxiv_07.tex
\section{Discussion and Conclusion}\label{sec:conclusion}
We investigate two relaxations of envy-freeness, hidden envy-freeness and envy
freeness up to any item, on graph structures. The relaxation to graphs is natural --- many real life agents only care about the agents with whom they interact. In many cases, it enables us to obtain results that have not been possible on complete graphs. Obtaining positive results on natural classes of graphs may also help prove the existence of EFX and HEF allocations more broadly. Our empirical results show that $G$-EFX allocations are likely to exist for several more general classes of graphs like paths $P_n$. It would be interesting to conclude whether the existence of a $G$-EFX allocation implies the existence of a $G'$-EFX allocation, where $G'$ is a subgraph of $G$, on the same set of goods. We would also like to know if any of the suggested potential functions are monotonic for arbitrary sequences of cutters and choosers, or if there are other potential functions that are more suitable. Several other notions of fairness like local proportionality and local max-min share are definable naturally on graphs, and these offer scope for future research.

%% file: arxiv_appendix.tex
\appendix
\newpage
\onecolumn
\section{Proof from Section \ref{sec:intro}}\label{apdx:intro}


\begin{proposition}
There are allocations that are EFX but not uHEF-$(n - 1)$, as well as allocations that are uHEF-$1$ but not EFX.
\end{proposition}
\begin{proof}
Consider three agents $1, 2, 3$, with three goods $\{g_1, g_2, g_3\}$, and consider the allocation $X$ that gives agent $i$ good $g_i$. Suppose $g_1 \succ_2 g_2 \succ_2 g_3$, $g_2 \succ_3 g_3 \succ_3 g_1$, and $g_3 \succ_1 g_1 \succ_1 g_2$. Then, $X$ is trivially EFX, but not uHEF-$2$; any visible good $g_i$ is envied by agent $(i + 1)\pmod 3$, and so all goods need to be hidden in order to eliminate envy.

Consider two agents $1$ and $2$ (with identical valuations), and three goods $(g_1, g_2, g_3)$ with values $(10, 2, 1)$ respectively. Suppose $X$ is the allocation $(\{g_1, g_3\}, \{g_2\})$. Then, $X$ is uHEF-$1$, as hiding $g_1$ eliminates all envy. However, it is not EFX, as $(X_1 - g_3) \succ_2 X_2$.
\end{proof}

\section{Proofs from Section \ref{sec:hidden}}\label{apdx:hidden}


\thmhiddenvertexcover*
\begin{proof}
Let the output allocation of \textsc{Vertex Cover Round Robin} be $X$ with hidden goods $S$. The algorithm hides exactly one good from every agent in $C$. Therefore, to show that it is $G$-uHEF-$k$, it suffices to show that for all edges $(i, i') \in E(G)$, we have $X_i \succ_i (X_{i'} \setminus S)$.

Consider any agent $i \in N$ in a connected component other than $G_j$. Neither $i$ nor any of its neighbors receive any goods. Thus, the hidden envy constraint is maintained trivially on all such components. We therefore only need to consider $G_j$, and show that no edge in $E(G_j)$ has envy in any direction. Consider any $(i, i') \in E(G_j)$. Since $C$ is a vertex cover of $G_j$, either $i \in C$ or $i' \in C$.

Suppose without loss of generality that $i' \in C$. \textsc{Vertex Cover Round Robin} allocates exactly one good to each agent in every round till there are no goods left.

\textbf{Case 1: $i \notin C$.} There is trivially no envy from $i'$ towards $i$, as $i'$ chose before $i$ in every round. Let us verify whether $i$ can envy $i'$. In the final round, if $i'$ does not receive a good, then $i$ does not either. This implies $|X_i| \ge |X_{i'}| - 1$. Consider the case when $|X_i| = |X_{i'}| - 1$ (the other case is easier, and omitted). Suppose $X_i = \{g_1, g_2, \ldots, g_{\ell}\}$ and $X_{i'} = \{g'_1, g'_2, \ldots, g'_{\ell+1}\}$ where $g_t, g'_t$ are the goods received by $i$ and ${i'}$ respectively in round $t$. Since $i$ chose $g_t$ over $g'_{t+1}$ in round $t$, we have $g_t \succ_i g'_{t+1}$ for each $t \in [\ell]$. Therefore, $X_i \succ_i X_{i'} - g'_1$. Since $g'_1$ is the first good given to $i'$, and $i' \in C$, the algorithm hides $g'_1$. Therefore, there is no envy from $i$ towards $i'$. 

\textbf{Case 2: $i \in C$.} When both $i$ and $i'$ are in $C$, then $g_1$ and $g'_1$ are both hidden. By the same argument as before, we have $X_i \succ_i X_{i'} - g'_1$, and $X_{i'} \succ_{i'} X_i - g_1$. It follows that there is no envy along this edge after $S$ is hidden.
\end{proof}

\thmhiddentight*
\begin{proof}
Let $n = |N|$. We construct an instance over this graph with $n^3$ goods $M = \{g_1, g_2, \ldots, g_{n^3}\}$. 

We give all agents in $N$ identical additive valuations (denoted by $v$) where $v(g_j) = 1 + 2^{-j}$ for every $g_j \in M$.
Trivially, for any $S \subseteq M$, $|S| < v(S)$. On the other hand, for any $S \subseteq M$, we have 
\begin{align*}
v(S) = |S| + \sum_{g_j \in S} 2^{-j} < |S| + \sum_{j=1}^{\infty} 2^{-j} = |S| + 1
\end{align*}

This is summarized in the following observation, that we shall use.

\begin{observation}\label{obs:vS-bounds}
For any $S \subseteq M$, $|S| < v(S) < |S|+1$.
\end{observation}

Now, consider any two non-empty sets $\varnothing \subsetneq S, T \subseteq M$. If $|S| \ne |T|$, from Observation \ref{obs:vS-bounds}, we have $v(S) \ne v(T)$. If $|S| = |T|$, let $g_j$ be the good with the lowest index $j$, in $S \cup T$, and assume WLOG $g_j \in S$. We have, using Observation \ref{obs:vS-bounds},

\begin{align*}
    v(S) \ge |S| + 2^{-j} \ge |S| + \sum_{j' = j+1}^{\infty} 2^{-j'} > |S| + \sum_{g_{j'} \in T} 2^{-j'} = |T| + \sum_{g_{j'} \in T} 2^{-j'} = v(T).
\end{align*}
This gives us our next observation.

\begin{observation}\label{obs:v-degenerate}
For any $\varnothing \subsetneq S, T \subseteq M$, if $S \ne T$, we have $v(S) \ne v(T)$.
\end{observation}

Now, coming back to the instance defined on $G$, we show that any $G$-HEF allocation hides at least $k$ goods. Let $X$ be any $G$-HEF allocation. There are $n^3$ goods and at most $n$ connected components, and so by the pigeonhole principle, some connected component has at least $n^2$ goods allocated. Let this component be $G'$. We have two possible cases. 

\textbf{Case 1: all agents in $G'$ receive at least 1 good.}
In this case, from Observation \ref{obs:v-degenerate}, there is envy along every edge of the connected component $G'$. This envy can only be eliminated by hiding a good from one of the endpoints of the edge in question. It follows that the set of agents who hide a good in $X$ must form a vertex cover of the component $G'$. Therefore, the number of hidden goods is at least $k$.

\textbf{Case 2: an agent $i \in V(G')$ receives $0$ goods.} In this case, since $G'$ has at least $n^2$ goods, by the pigeonhole principle, there must be some agent $j$ in $G'$ with at least $n$ goods. Since $G'$ is a connected component, there is a path from $i$ to $j$. Let this path be $(i_1, i_2, \ldots, i_t)$ where $i_1 = i$ and $i_t = j$.

If $|X_{i_\ell}| < |X_{i_{\ell+1}}|$, we must hide at least $|X_{i_{\ell+1}}| - |X_{i_{\ell}}|$ goods in $X_{\ell+1}$ to eliminate the envy between $i_\ell$ and $i_{\ell+1}$ (Observation \ref{obs:vS-bounds}). We therefore have that the number of hidden goods in $X$ is at least 
\begin{align*}
    \sum_{\ell \in [t-1]} \max\{|X_{i_{\ell+1}}| - |X_{i_{\ell}}|, 0\} \ge \sum_{\ell \in [t-1]} |X_{i_{\ell+1}}| - |X_{i_{\ell}}| = |X_j| - |X_i| = n,
\end{align*}
and therefore, at least $n$ goods need to be hidden in the allocation $X$. Since $k$ is trivially upper bounded by $n$, the proof is complete.
\end{proof}

\corhiddenhardness*
\begin{proof}
We show part (a) here. Part (b) can be shown similarly. 

Our reduction, unsurprisingly, is from \textsc{Vertex Cover}: given an undirected graph $G = (N, E)$ and positive integer $K$, the \textsc{Vertex Cover} problem simply asks if $G$ has a vertex cover of size at most $K$.
Given an instance $(G, K)$ of the vertex cover problem, we create an instance of the graphical fair allocation problem with $|N|$ agents and $n^3$ goods as described in Theorem \ref{thm:hiddentight}. Recall that all the agents have identical additive valuations.
Note that this instance can be built easily in polynomial time with respect to the original instance $(G, K)$.
We show that a vertex cover of size at most $K$ exists in the orginal problem if and only if a $G$-HEF-$K$ allocation exists in the fair allocation instance. 

Assume there exists a vertex cover of size $k \le K$ in the orginal vertex cover problem. Using Theorem \ref{thm:hiddenvertexcover}, there exists a $G$-HEF-$k$ allocation in the fair allocation instance. Since $k \le K$, we can conclude that there exists a $G$-HEF-$K$ allocation.

Assume the minimum vertex cover of graph $G$ has size $k > K$. Then, by Theorem \ref{thm:hiddentight}, for our specific fair allocation instance, there exists no allocation with at most $K$ hidden goods. This completes the reduction.
\end{proof}

\hiddennottight*
\begin{proof}
Consider the following graph $G$, with $n \ge 3$ nodes. We have a star, $K_{1, \ceil{\frac{n-1}{2}}-1}$, with the central node labeled $1$ and the outer nodes labeled $2, 3, \ldots, \ceil{\frac{n-1}{2}}$. The center of the star is also connected to node $\ceil{\frac{n-1}{2}}+1$, which is connected to a single node in a clique $K_{\floor{\frac{n-1}{2}}}$. The graph is shown in Figure~\ref{fig:vcnotrequired}.
When $n \ge 3$, the clique $K_{\floor{\frac{n-1}{2}}}$ is guaranteed to be non-empty. 
Of course, any vertex cover of $G$ needs to have $\Theta(n)$ vertices, as the clique itself needs all of its vertices (except for one) to be covered.

\input{Figures/vc-not-required}

To compute a $G$-uHEF-$2$ allocation with $n$ goods, we run a picking sequence protocol with the order $1, 2, 3, \ldots \ceil{\frac{n-1}{2}}, \ceil{\frac{n-1}{2}} + 1, 1, 2, \ldots \floor{\frac{n-1}{2}}$. 
In the protocol, we iterate through the agents in the specified order, and give each agent their most preferred good among all currently unallocated goods.
In such an allocation, only two goods need to be hidden to achieve envy freeness:
\begin{inparaenum}[(a)]
    \item the first good assigned to $1$, and
    \item the sole good assigned to $\ceil{\frac{n-1}{2}} + 1$. 
\end{inparaenum}

Once these two goods are hidden, using an argument similar to that of Theorem \ref{thm:hiddenvertexcover}, we can show that no agent adjacent to agent $1$ envies it. Agent $1$ does not envy any of its neighbors as well. 
Moving on to the clique, none of the agents in the clique receive any good, so they do not envy each other. For the same reason, agent $\ceil{\frac{n-1}{2}}+1$ does not envy any agent in the clique. The agents in the clique do not envy agent $\ceil{\frac{n-1}{2}}+1$ either since the entire bundle allocated to $\ceil{\frac{n-1}{2}}+1$ is hidden.
Therefore, the allocation is $G$-uHEF-$2$.
\end{proof}

\section{Proofs from Section \ref{sec:theory}}\label{apdx:theory}


\propgefxstar*

\begin{proof}
Let the center of the star correspond to agent $c \in N$. If all the outer agents had the same valuation function $v_c$ as $c$, a complete EFX allocation would be guaranteed to exist \citep{efxcanon}. Let $Y = (Y_1, Y_2, \ldots, Y_n)$ be one such EFX allocation where all the agents have the valuation function $v_c$.

We iterate through the outer vertices (in any order) and construct an allocation $X$ by allocating to each outer vertex its highest-valued bundle in $Y$ that has not already been allocated. We allocate the final bundle to the center of the star, i.e. to agent $c$.

We claim that the allocation $X$ is $G$-EFX. If $c'$ is an outer vertex, $c'$ does not envy $c$, as they picked their bundle $X_{c'}$ before the bundle $X_c$ was chosen. The center, $c$, does not strongly envy any of the outer vertices, because $Y$ is EFX for agents with the valuation function $v_c$. Otherwise, we would have strong envy between the center and some outer vertex, say $c'$, implying that for some good $g \in X_{c'}$, we have $X_{c'} - g \succ_c X_c$. However, this implies the existence of strong envy in the allocation $Y$, which is a contradiction.
\end{proof}

\thmresultone*

\begin{proof}
The proof is similar to that of Proposition \ref{prop:gefx:star}. Suppose the core vertices all have identical valuation function $v_c$, and let $Y = (Y_1, Y_2, \ldots, Y_n)$ be an EFX allocation for this problem instance assuming that all $n$ agents have this same valuation function $v_c$.

We construct an allocation $X$ as follows. We first iterate through the agents in $N\setminus N'$ (in any order) and allocate to each such agent its highest-valued bundle in $Y$ that has not already been allocated. Once this is done, we distribute the remaining bundles in $Y$ to the agents in $N'$ (in any order).

We claim that the allocation $X$ is $G$-EFX. If $i \in N\setminus N'$, then $i$ has neighbors only in $N'$. Clearly, $i$ does not envy any vertex in $N'$, because they picked their bundle $X_i$ over each of the bundles distributed among $N'$. If $i \in N'$, then $i$ has valuation function $v_c$ by definition. Of course, $i$ does not strongly envy any of the other vertices because $Y$ is EFX for agents with the valuation function $v_c$. 
Otherwise, we would have strong envy between $i$ and some other vertex, say $i'$, implying that for some good $g \in X_{i'}$, we have $X_{i'} - g \succ_i X_i$. However, this implies the existence of strong envy in the allocation $Y$, which is a contradiction.
\end{proof}

\resultonecorollary*
\begin{proof}
Theorem \ref{thm:result-1} implies the existence of a $G$-EFX allocation on the graph in Figure \ref{subfig:star-general-1}. The EFX criterion is maintained between all edges in the graph, and therefore between all pairs of agents except possibly the one consisting of the top and bottom vertices.
\end{proof}

\thmresulttwo*

\begin{proof}
If all agents have consistent additive valuations, an EFX allocation is guaranteed to exist \citep{efxcanon}.
Let $Y = (Y_1, \ldots, Y_n)$ be an EFX allocation for this problem instance assuming all the core vertices in $N'$ have their respective valuation functions, and each vertex in $N\setminus N'$ has an identical valuation function to any of its neighbors (this is well-defined, by construction).

We construct an allocation $X$ as follows. We first iterate through the agents in $N\setminus N'$ (in any order). For each such agent $i$, suppose it is connected to agents in $N'_k$. We let $i$ choose its highest-valued bundle in $Y$ allocated to any agent with valuation $v_k$ and not yet chosen by any other agent. Once this is done for all $i \in N\setminus N'$, we distribute the remaining bundles in $Y$ to the agents in $N'$ while maintaining the invariant that any such agent $i$ receives a bundle allocated in $Y$ to an agent with valuation function $v_i$.

We claim that this algorithm terminates with a bundle to each agent in $N$. There is no point in the algorithm where we wish to assign a bundle in $Y$ to an agent $i$, but all bundles in $Y$ allocated to agents with valuation $v_i$ have already been assigned. This follows by construction of $Y$, and the fact that each agent in $N \setminus N'$ only selects from bundles intended for neighbors of agents in $N'_k$.

We further claim that the allocation $X$ is $G$-EFX. If $i \in N\setminus N'$, with neighborhood $\mathrm{Nbd}_G(i) \subseteq N'_k$, then $i$ is allocated their bundle from $Y$ before any node of $N'_k$, and they are all allocated from the same pool of bundles (corresponding to agents with valuation $v_k$). Therefore, $i$ does not envy any agent in $N'_k$. If $i \in N'$, they do not strongly envy their neighbors as that would violate the EFX property for the allocation $Y$ for similar reasons as in the proof of Theorem \ref{thm:result-1}. 
\end{proof}

\thmtimecomplexityone*

\begin{proof}
An EFX allocation can be computed\footnote{ This assumes the goods are sorted in order of any of the consistent valuations, so the result is up to an additive factor of $O(m\log m)$. We will ignore this factor for our analysis, as it can be achieved by a separate pre-processing step to sort the goods by the core vertex valuations.} for agents with consistent additive valuations in $O(m n^3)$ time \citep{efxcanon}. The initial EFX allocation $Y$ in Theorem \ref{thm:result-1} and \ref{thm:result-2} can be computed using this algorithm. The next step involves iterating through the agents and giving them their best unallocated bundle that satisfies certain additional conditions. For each agent, we can find this bundle in $O(nm)$ time since there are only at most $n$ bundles and computing the valuation of each bundle can be done in $O(m)$ time (using additivity). Therefore the second step of the algorithm takes $O(m n^2)$ time. This gives us an overall time complexity of $O(m n^3)$.

\end{proof}

\threeedgepath*
\begin{proof}
Let the path have vertices $1, 2, 3, 4$ in that order, with valuation functions $v_1, v_2, v_3, v_4$ (see Figure \ref{fig:three-edge-path}). We will use the fact that an EFX allocation is guaranteed to exist when each agent only has one of two types of (general) valuations \citep{twotypes}. Let $Y = (Y_1, Y_2, Y_3, Y_4)$ be an EFX allocation of the set of goods $M$ on four agents with valuations $(v_2, v_2, v_3, v_3)$ respectively.

We construct an allocation $X$ from $Y$ as follows. We allocate to agent $1$ their highest-valued bundle in the set $\{Y_1, Y_2\}$, and assign the other bundle in that set to agent $2$. Similarly, we allocate to agent $4$ their highest-valued bundle in the set $\{Y_3, Y_4\}$, and assign the other one to agent $3$.

We claim the allocation $X$ is $G$-EFX on the path $P_4$. Agents $1$ and $4$ do not envy agents $2$ and $3$ respectively since they were allocated a bundle that they (weakly) prefer to that of their unique neighbor. Agents $2$ and $3$ do not strongly envy any other agent, because $Y$ is EFX for agents with the valuation functions $v_2$ or $v_3$.
\end{proof}

\thmresultthree*
\begin{proof}
Once again, we will use the fact that an EFX allocation is guaranteed to exist when each agent only has one of two types of valuations \citep{twotypes}. Consider a modified instance of the problem, on the same graph, but where all the outer vertices in $G$ have the same valuation function as their neighbors among the core vertices. Note that this is well-defined by construction, and furthermore, this instance has agents with only types $v_k$ and $v_\ell$. So let $Y = (Y_1, \ldots, Y_n)$ be an EFX allocation for this modified instance.

We first divide $Y$ into two {\em pools} of bundles based on the valuation function of the agent they were allocated to. Suppose $Y^k$ is the set of bundles allocated in $Y$ to agents with valuation $v_k$, and $Y^\ell$ is the set of bundles allocated to agents with valuation $v_\ell$. 

We construct an allocation $X$ by allocating the bundles in $Y$ in a particular order. We start with the outer agents in $N\setminus N'$ whose neighborhood is contained in $N'_k$. We iterate through these agents (in any order), allocating to each such agent their highest-valued bundle in $Y^k$ that has not been allocated yet. Then, we assign the remaining bundles in $Y^k$ in any order to the agents in $N'_k$. We repeat this same procedure with the remaining agents and the set of bundles $Y^\ell$, starting with the outer agents with neighborhoods in $N'_\ell$, as before.

We claim that this algorithm terminates with a bundle to each agent in $N$. This follows by similar arguments as in the proof of Theorem \ref{thm:result-2}. We also claim that the allocation $X$ is $G$-EFX. If $i \in N\setminus N'$, then $i$ is allocated their bundle from the same pool of bundles as all their neighbors, but before any of their neighbors are. So $i$ does not envy any of their neighbors. If $i \in N'$, they do not strongly envy their neighbors as that would violate the allocation $Y$ being EFX, by a similar argument as in the proof of Theorem \ref{thm:result-2}.
\end{proof}

\lexicographic*

\begin{proof}
Denote the distance function in $G$ by $\mathrm{dist}_G(i, j)$. Choose a pair of agents $u$ and $v$ with $\mathrm{dist}_G(u, v) \geq 4$. We will find a subset $S \subseteq N$ to disconnect $u$ and $v$. Then we will assign goods to the component of $G[N\setminus U]$ containing $u$, chores to the component of $G[N\setminus U]$ containing $v$, and nothing (the empty allocation) to all other agents.

Let $S = \{i \in N : \mathrm{dist}_G(i, u) \geq 2 \land \mathrm{dist}_G(i, v) \geq 2\}$. Note that $N \setminus S$ consists of $u$, $v$, and all neighbors of $u$ and $v$. $S$ is non-empty, since the shortest path between $u$ and $v$ contains at least one agent at distance exactly $2$ from both. Also, $\mathrm{Nbd}_G(u) \cap \mathrm{Nbd}_G(v) = \varnothing$, as otherwise $\mathrm{dist}_G(u, v) \leq 3$.

Let the neighbors of $u$ select their (single) highest-priority good in any order, from the unallocated goods. Let the number of neighbors of $v$ be $n_v$. Assign the $n_v$ chores with the highest priority for $v$ to its neighbors in any order, one to each. Assign the remaining goods (if any) to $u$, and the remaining chores (if any) to $v$. If we do not have enough neighbors to account for all the goods (or chores), the result is trivial.

Any envy in $G$ directed towards any neighbor of $u$ is not strong, since this neighbor has at most one good. The only envy directed towards $u$ can come from its neighbors, but this does not happen, since they all selected their top-priority goods before $u$ was given its bundle (using the lexicographic property).

The neighbors of $v$ in $G$ do not have strong envy towards any other agent, since each of them has at most one chore. The only possibly envy towards these neighbors is from agent $v$, but this is not strong, since we assigned $v$'s top-priority chores to its neighbors before assigning any chores to $v$ (again, using the lexicographic property).

There cannot be strong envy anywhere else in the graph $G$.

Note that this allocation can be computed in polynomial time. We can find the pair of vertices $u$ and $v$ in $\Theta(n^3)$ time using the Floyd-Warshall algorithm; selecting and assigning the top-priority goods and chores for $u$ and $v$ to their neighborhoods takes time $\Theta(nm)$ for instances with $m$ items.





\end{proof}

%% file: Figures/vc-not-required.tex
\begin{figure*}
    \centering
    
    \begin{tikzpicture}[node distance = 1.25cm]

    \def \n {20}
    \def \N {8}
    \def \radius {2cm}
    \def \rd {1mm}
    \def \rer {4mm}
    
    \def \margin {8} 
    
    \node[draw, circle] at (360:0mm) (center) [label={[label distance=1mm]180:1}]{};
    \foreach \i [count=\ni from 0] in {4, 3, 2}{
      \node[draw, circle] at ({135-\ni*45}:\radius) (\ni) [label={[label distance=0.5mm]90:$\i$}]{};
      \draw (center)--(\ni);
    }
    
    \foreach \i [count=\ni from 3] in {\ceil{\frac{n-1}{2}}+1, \ceil{\frac{n-1}{2}}}{
      \node[draw, circle] at ({135-\ni*45}:\radius) (\ni) [label={[label distance=0.5mm]270:$\i$}]{};
      \draw (center)--(\ni);
    }
    
    \draw[dotted] (-78:\radius) arc[start angle=-78, end angle=-202, radius=\radius];
    
    \node[draw, fill=lightgray, circle, minimum size = 4cm] at (360:3*\radius) (bridge) {\Huge$K_{\floor{\frac{n-1}{2}}}$};
    \draw (3)--(bridge);
    
    
    \end{tikzpicture}
    
    \caption{A graph with $n$ vertices in which any instance with $n$ goods admits a $G$-uHEF-$2$ allocation, under additive valuations. Any vertex cover of this graph is of size $\Theta(n)$}
    \label{fig:vcnotrequired}
\end{figure*}
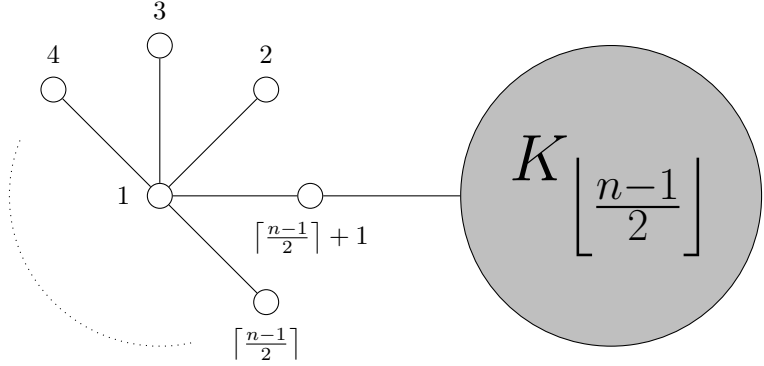